\documentclass[journal,draftcls,onecolumn,12pt,twoside]{IEEEtran}
\usepackage[T1]{fontenc}% optional T1 font encodin
\usepackage{ifpdf}

\usepackage{cite}

\usepackage[dvips]{graphicx}
\graphicspath{{../}}
\usepackage{geometry}
\usepackage{subfig}
\usepackage{caption}
\DeclareGraphicsExtensions{.pdf,.jpeg,.png,.eps}

\usepackage{amsmath}
\usepackage{amsfonts}
\usepackage{amssymb}
\usepackage{mathrsfs}
\usepackage[cmintegrals]{newtxmath}
\usepackage{bm}

\newtheorem{lemma}{\textbf{Lemma}}
\newtheorem{theorem}{\textbf{Theorem}}

\newtheorem{corollary}{\textbf{Corollary}}

\newenvironment{proof}{{\noindent\it Proof:}}{\hfill $\square$}
\usepackage[ruled,linesnumbered]{algorithm2e}
\usepackage{array}
\usepackage{multirow}
\usepackage[table,xcdraw]{xcolor}
\usepackage{booktabs,threeparttable}

\usepackage{url}
\usepackage{algorithm2e}

\begin{document}
\title{A Balanced Tree Approach to Construction of Length-Flexible Polar Codes}
\author{Xinyuanmeng~Yao and Xiao~Ma~\IEEEmembership{Member,~IEEE}
\thanks{The authors are with the School of Computer Science and Engineering, and also with Guangdong Province Key Laboratory of Information Security Technology, Sun Yat-sen University, Guangzhou 510006, China 
(e-mail: yaoxym@mail2.sysu.edu.cn, maxiao@mail.sysu.edu.cn).}}
\maketitle
\begin{abstract}
From the perspective of tree, we design a length-flexible coding scheme. For an arbitrary code length, we first construct a balanced binary tree~(BBT) where the root node represents a transmitted codeword, the leaf nodes represent either active bits or frozen bits, and a parent node is related to its child nodes by a length-adaptive $(U+V\mid V)$ operation. Both the encoding and the successive cancellation~(SC)-based decoding can be implemented over the constructed coding tree. For code construction, we propose a signal-to-noise ratio~(SNR)-dependent method and two SNR-independent methods, all of which evaluate the reliabilities of leaf nodes and then select the most reliable leaf nodes as the active nodes. Numerical results demonstrate that our proposed codes can have comparable performance to the 5G polar codes. To reduce the decoding latency, we propose a partitioned successive cancellation~(PSC)-based decoding algorithm, which can be implemented over a sub-tree obtained by pruning the coding tree. Numerical results show that the PSC-based decoding can achieve similar performance to the conventional SC-based decoding. 
\end{abstract}

\section{Introduction}
As a class of error-correcting codes which can provably achieve capacity for binary-input discrete memoryless channels~(B-DMCs) with low encoding and decoding complexity~\cite{Arikan2009}, polar codes have attracted a lot of attention in the past decade. Recently, they have been adopted as the channel coding scheme for the control channel in the Enhanced Mobile Broadband scenario in the fifth generation~(5G) wireless communications standard~\cite{3GPP38212}. For polar coding, one of the most important techniques is channel polarization. Based on a mapping from two independent and identical B-DMCs to a ``worse'' channel and a ``better'' channel, $2^n$ copies of an original B-DMC are transformed into $2^n$ bit-channels with different reliabilities. Then, the most reliable polarized channels are chosen to transmit data bits. In the finite-length regime, polar codes with the successive cancellation list~(SCL) decoder or the cyclic redundancy check-aided SCL~(CA-SCL) decoder can have satisfactory performance. 

From the definition of polar codes, we see that the code rates of polar codes are flexible. However, the lengths of polar codes are limited to powers of two, which are not convenient for practical applications. To construct length-flexible polar codes, one can use puncturing, shortening or extending. The punctured and shortened~(P/S) polar codes~\cite{Eslami2011,Niu2013,Shin2013,Wang2014,Miloslavskaya2015,Jang2019,Zhao2021,Han2022} are obtained by deleting some coded bits of their mother codes, while the extended polar codes~\cite{Chen2013,Saber2015,Ma2017,Zhao2018,Jang2020} are obtained by padding some bits to slightly shorter polar codes, which are usually designed for hybrid automatic repeat-and-request. It has been shown in the existing literature that the P/S polar codes can have good performance. However, the decoding of the P/S polar codes is based on their mother codes, which results in an increased decoding latency. For some parameter configurations~(for example, the cases when the code length is slightly larger than a power of two and the code rate is low), the extended polar codes might be a better choice in terms of performance and complexity. Recently, for 5G applications, a well designed low-complexity rate matching scheme has been proposed in~\cite{3GPP38212}, which incorporates puncturing, shortening and repetition~(the simplest way of extending) techniques. In fact, one can also construct length-flexible polar codes with decoding latency depending only on the real code lengths~(in contrast to the P/S polar codes depending on their mother codes' lengths). Multi-kernel polar coding~\cite{Bioglio2020} is one of such techniques, which mixes kernels of different sizes. Notice that any positive integer can be decomposed into a sum of powers of two, chained polar subcodes have been proposed in~\cite{Trifonov2018}. Another technique is asymmetric polar coding~\cite{Cavatassi2019}, which links polar codes of unequal lengths.

Notice that the conventional polar coding~\cite{Arikan2009} can be implemented over a full binary tree where the root node represents a transmitted codeword, each of the leaf nodes represents either an active bit or a frozen bit and the relationship between a parent node and its child nodes is specified by the $(U+V\mid V)$ construction. Motivated by this idea, we propose to construct a balanced binary tree~(BBT) for any code length, referred to as a coding tree, and design a length-flexible coding scheme by implementing a length-adaptive $(U+V\mid V)$ operation over the coding tree. For BBT polar code construction, we turn to density evolution with Gaussian approximation~(GA), termed as GA construction. We also present two signal-to-noise ratio~(SNR)-independent methods for practical use. Similar to~\cite{Niu2019,Niu2020,Niu2021}, we assign to each node a polar subcode and present a recursive algorithm to approximately calculate the minimum Hamming weight~(MHW) and the multiplicity of the MHW codewords for a polar subcode. Then, a theoretically explainable MHW construction is proposed. In addition, based on~\cite{He2017}, we define a polarization weight~(PW) for the leaf nodes in the coding tree and propose a heuristic PW construction. 

From the viewpoint of coding tree, the SC decoder for~(BBT) polar codes makes hard decisions at leaf nodes in a serial manner, which results in high decoding latency when the code length is long. For the conventional polar codes, the existing works~\cite{Alamdar-Yazdi2011,Sarkis2014,Hanif2017,Hashemi2016,Hashemi2017,Hanif2018,Condo2018} identified different types of internal nodes according to the pattern of active and frozen bits and then designed the corresponding multi-bits parallel decoders. For example, Alamdar-Yazdi~\emph{et al.} proposed to make hard decisions at rate-0 and rate-1 codes~\cite{Alamdar-Yazdi2011}, Sarkis~\emph{et al.} further proposed to make hard decisions at single-parity-check and repetition codes~\cite{Sarkis2014}, and Hanif~\emph{et al.} proposed to make hard decisions at five types of codes~(Type-\uppercase\expandafter{\romannumeral1}, Type-\uppercase\expandafter{\romannumeral2}, Type-\uppercase\expandafter{\romannumeral3}, Type-\uppercase\expandafter{\romannumeral4} and Type-\uppercase\expandafter{\romannumeral5} codes)~\cite{Hanif2017}. Recently, for the 5G polar codes involving parity-check~(PC) bits, Zhou~\emph{et al.} identified various types of special nodes which contain PC bits and proposed the corresponding fast decoders~\cite{Zhou2023}. In this work, we present a partitioned successive cancellation~(PSC) decoder. First, a decoding sub-tree is extracted from the coding tree according to a preset threshold on the dimensions of the nodes and then the PSC decoder is implemented over the decoding sub-tree. To analyze the error performance of the BBT polar codes with the PSC decoding, we derive three bounds on the frame error rate~(FER). Furthermore, the PSC decoder can be adapted to the partitioned successive cancellation list~(PSCL) decoder or CA-PSCL decoder for the BBT polar codes.

The rest of this paper is organized as follows. We introduce the coding tree and present our BBT polar coding scheme in Section~\ref{section2}. Three BBT polar code constructions and their performance are presented in Section~\ref{section3}. In Section~\ref{section4}, the proposed PSC-based decoding is described and evaluated in terms of error-correction performance and decoding latency. We conclude this paper in Section~\ref{section5}.

\section{The Length-Flexible Coding Scheme}\label{section2}
In this section, we first introduce a coding tree and then present a coding scheme which is applicable to arbitrary code length.
\subsection{Coding Tree}\label{section2subsection1}
The dynamic behavior of the encoder and the decoder in our scheme can be graphically represented by a rooted BBT, called a coding tree, as defined below. Given a code length $N$, the root node is labeled by a codeword $\boldsymbol{c}\in\mathbb{F}_2^N$~(the $N$-dimensional vector space over the binary field $\mathbb{F}_2$). Each node labeled by a vector $\boldsymbol{v}\in\mathbb{F}_2^{\ell}$, where $2\leq \ell\leq N$, has a left child node labeled by $\boldsymbol{v}_{l}\in\mathbb{F}_2^{\lceil \ell/2\rceil}$ and a right child node labeled by $\boldsymbol{v}_{r}\in\mathbb{F}_2^{\lfloor \ell/2\rfloor}$, where $\lceil\cdot\rceil$ and $\lfloor\cdot\rfloor$ denote the ceiling function and the floor function, respectively. For convenience, we say that a node labeled with $\boldsymbol{v}\in\mathbb{F}_2^{\ell}$ has a length $\ell$. If the lengths of two children are equal, i.e., $\lceil \ell/2\rceil=\lfloor \ell/2\rfloor$, define as usual 
\begin{align}\label{Equation1}
\boldsymbol{v}_{l}\oplus \boldsymbol{v}_{r}\triangleq(v_{l,0}+v_{r,0},\ldots,v_{l,\lceil \ell/2\rceil-1}+v_{r,\lfloor \ell/2\rfloor-1}),
\end{align}
and if the length of the left child is one greater than that of the right child, i.e., $\lceil \ell/2\rceil=\lfloor \ell/2\rfloor+1$, define
\begin{align}\label{Equation2}
\boldsymbol{v}_{l}\oplus \boldsymbol{v}_{r}\triangleq(v_{l,0}+v_{r,0},\ldots,v_{l,\lceil  \ell/2\rceil-2}+v_{r,\lfloor \ell/2\rfloor-1},v_{l,\lceil  \ell/2\rceil-1}),
\end{align}
where ``$+$'' in the right hand sides of~(\ref{Equation1}) and~(\ref{Equation2}) denote the addition over $\mathbb{F}_2$. Then, the relationship of the parent node and its children is given by
\begin{equation}\label{Equation3}
\boldsymbol{v}=(\boldsymbol{v}_{l}\oplus \boldsymbol{v}_{r},\boldsymbol{v}_{r}).
\end{equation}
In words, in the case when $\lceil \ell/2\rceil=\lfloor \ell/2\rfloor$, (\ref{Equation3}) is exactly the same as a polar transformation, while in the case when $\lceil \ell/2\rceil=\lfloor \ell/2\rfloor+1$, (\ref{Equation3}) is equivalent to extending the right child $\boldsymbol{v}_{r}$ by one zero for vector addition and then shortening it from $\boldsymbol{v}$ and $\boldsymbol{v}_{r}$. See Fig.~\ref{vector_addition} for reference. 

\begin{figure}[t]  
	\centering
	\subfloat[$\lceil\ell/2\rceil = \lfloor \ell/2\rfloor$]{
		\includegraphics[width=0.5\textwidth]{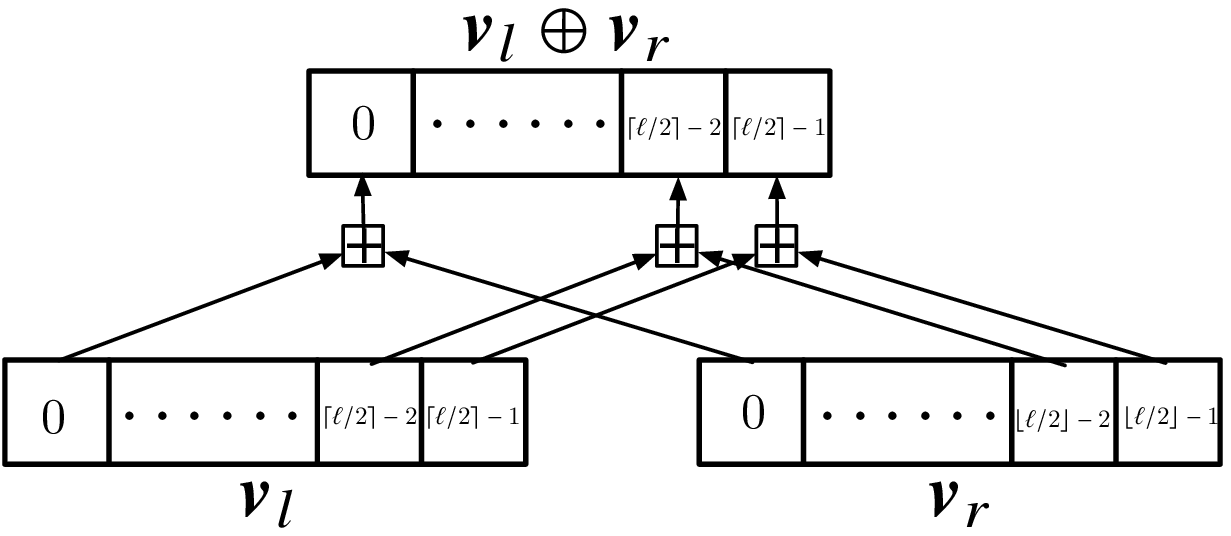}
	}
	\subfloat[$\lceil\ell/2\rceil = \lfloor \ell/2\rfloor+1$]{
		\includegraphics[width=0.5\textwidth]{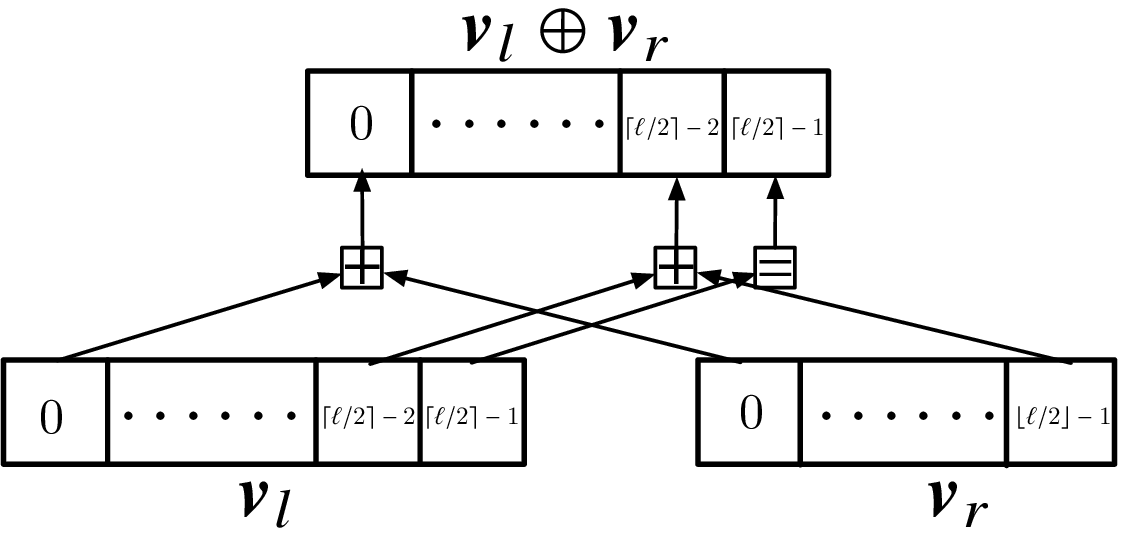}
	}
	\caption{Illustration of the definition of ``$\boldsymbol{v}_{l}\oplus \boldsymbol{v}_{r}$'', where the lengths of $\boldsymbol{v}_{l}\oplus \boldsymbol{v}_{r}$, $\boldsymbol{v}_{l}$ and $\boldsymbol{v}_{r}$ are $\lceil \ell/2 \rceil$, $\lceil \ell/2 \rceil$ and $\lfloor \ell/2 \rfloor$, respectively. }\label{vector_addition}
\end{figure}

\begin{figure}[t]
	\centering
	\includegraphics[width=0.35\textwidth]{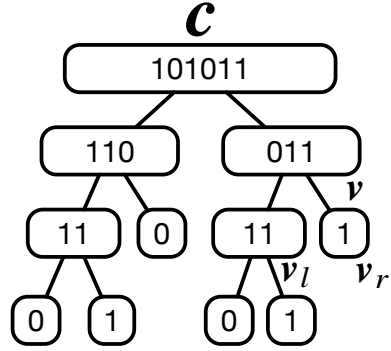}
	\caption{The coding tree for $N=6$, where the labels of the leaf nodes from left to right are $0,1,0,0,1,1$ and the codeword of the root node is $101011$. }\label{coding_tree}
\end{figure}

With the above definition, for any code length $N$, we can first construct the corresponding coding tree from the root by splitting each node with length $\ell\geq 2$ into two children. The code tree has $\lceil \log_2 N\rceil$ levels and $N$ leaves. Then, given $N$ labels of the leaf nodes, we can calculate the labels of all internal nodes recursively from the leaves to the root. An example of a coding tree for code length $N=6$ is shown in Fig.~\ref{coding_tree}. The tree has three levels in total. Two leaf nodes are at the second-to-last level and four leaf nodes are at the last level. The labels of the leaf nodes from left to right are $0,1,0,0,1,1$ and the codeword of the root node is $101011$.

\subsection{Encoding and Decoding}\label{section2subsection2}
Denote by $K$ and $N$ the code dimension and the code length, respectively. The code rate is $R\triangleq K/N$. To encode a data sequence $\boldsymbol{u}\in\mathbb{F}_2^K$, we first construct the unlabeled coding tree for $N$ according to Section~\ref{section2subsection1} and distinguish the $N$ leaf nodes by $K$ active nodes and $N-K$ frozen nodes. Then the $K$ active nodes are labeled by the bits of $\boldsymbol{u}$ and the $N-K$ frozen nodes are labeled by zero. The choice of the active nodes will be discussed later. The encoder computes the labels of the non-leaf nodes in a reverse level-order traversal of the tree, resulting in a transmitted codeword $\boldsymbol{c}\in\mathbb{F}_2^N$ at the root node. 

The decoding can also be implemented over the coding tree. In this situation, each tree node is associated with an LLR vector and an HBE vector. Then, following three basic rules, the decoder computes the LLR vectors of the nodes in a pre-order traversal and computes the HBE vectors of the nodes in a post-order traversal, as detailed below. For an internal node  of length $\ell$ with LLRs $\boldsymbol{\alpha}\triangleq(\alpha_0,\alpha_1,\ldots,\alpha_{\ell-1})\in\mathbb{R}^{\ell}$ and HBEs $\boldsymbol{\beta}\triangleq(\beta_0,\beta_1,\ldots,\beta_{\ell-1})\in\mathbb{F}_2^{\ell}$, denote by $\boldsymbol{\alpha}_{l}\triangleq(\alpha_{l,0},\alpha_{l,1},\ldots,\alpha_{l,\lceil \ell/2\rceil-1})\in\mathbb{R}^{{\lceil \ell/2\rceil}}$ and $\boldsymbol{\beta}_{l}\triangleq(\beta_{l,0},\beta_{l,1},\ldots,\beta_{l,\lceil \ell/2\rceil-1})\in\mathbb{F}_2^{{\lceil \ell/2\rceil}}$ the LLRs and HBEs of its left child, respectively, and similarly by $\boldsymbol{\alpha}_{r}\triangleq(\alpha_{r,0},\alpha_{r,1},\ldots,\alpha_{r,\lfloor \ell/2\rfloor-1})\in\mathbb{R}^{\lfloor \ell/2\rfloor}$ and $\boldsymbol{\beta}_{r}\triangleq(\beta_{r,0},\beta_{r,1},\ldots,\beta_{r,\lfloor \ell/2\rfloor-1})\in\mathbb{F}_2^{\lfloor \ell/2\rfloor}$ the LLRs and HBEs of its right child, respectively. For convenience, define 
\begin{equation}\label{Equation4}
f(a,b)=\ln\frac{1+e^{a+b}}{e^a+e^b},\text{~for~}a,b\in\mathbb{R},
\end{equation} 
and
\begin{equation}\label{Equation5}
g(a,b,c)=b+(-1)^c\cdot a,\text{~for~}a,b\in\mathbb{R}\text{~and~}c\in\mathbb{F}_2 .
\end{equation}
The decoder starts by initializing the LLRs of the root node with a received vector $\boldsymbol{y}\in\mathbb{R}^N$ and then proceeds according to the following three rules.  
\begin{enumerate}
\item When the LLRs of a parent node are available, the LLRs of its left child are computed. 
If $\ell$ is even, 
\begin{equation}
	\alpha_{l,i}=f(\alpha_i,\alpha_{\lceil \ell/2\rceil+i}),\text{~for~}0\leq i\leq \lceil \ell/2\rceil-1,
\end{equation}
which is the same as the case in the SC decoding of polar codes. If $\ell$ is odd, 
\begin{align}
	\alpha_{l,i}=
	\begin{cases}
		\alpha_{l,i}=f(\alpha_i,\alpha_{\lceil \ell/2\rceil+i}),&\text{~for~}0\leq i\leq \lceil \ell/2\rceil-2\\
		\alpha_{l,i}=\alpha_{i},&\text{~for~}i=\lceil \ell/2\rceil-1
	\end{cases}.
\end{align}
\item When the LLRs of a parent node and the HBEs of its left child are available, the LLRs of its right child are computed as
\begin{equation}
	\alpha_{r,i}=g(\alpha_i,\alpha_{\lceil \ell/2\rceil+i},\beta_{l,i}),\text{~for~}i=0,1,\ldots,\lfloor \ell/2\rfloor-1.
\end{equation}
\item When the HBEs of a pair of sibling nodes are available, the HBEs of their parent 
are computed as
\begin{equation}
\boldsymbol{\beta}=(\boldsymbol{\beta}_l\oplus\boldsymbol{\beta}_r,\boldsymbol{\beta}_r),
\end{equation}
where the operation ``$\oplus$'' is defined in $(\ref{Equation1})$ and $(\ref{Equation2})$.
\end{enumerate}
Notice that, in contrast to the LLRs which are computed from the root nodes, the HBEs are computed from the leaf nodes in a serial manner. The leaf nodes are of two kinds:~the active nodes and the frozen nodes. For a frozen leaf node, its HBE is directly set to zero. For an active leaf node, its HBE is set to one if its LLR is less than zero, and zero otherwise.  Once all the tree nodes are estimated, the decoding process is terminated and we can extract an estimate of the data sequence from the active leaf nodes.

\subsection{Algebraic Description}\label{section2subsection3}
Given a coding tree and a choice of the active leaf nodes which is termed as rate-profiling, all possible labels of the root node form a binary linear code. We refer to such a code as a BBT polar code. The encoding process of a BBT polar code with length $N$ and dimension $K$ can be described as follows. We index the leaf nodes from left to right by $0,1,\ldots,N-1$. Then, denote by $\mathscr{A}$ and $\mathscr{F}\triangleq\mathscr{A}^c$ the index set of active leaf nodes and that of frozen leaf nodes, respectively, where $\mathscr{A}\subset\{0,1,\ldots,N-1\}$ and $|\mathscr{A}|=K$. Let $\boldsymbol{w}\in\mathbb{F}_2^N$ be a leaf node sequence such that $\boldsymbol{w}_{\mathscr{A}}$~(the components of $\boldsymbol{w}$ indexed by $\mathscr{A}$) is a data sequence and $\boldsymbol{w}_{\mathscr{F}}$ is the all-zero sequence. Then, we can obtain a codeword $\boldsymbol{c}=\boldsymbol{w}\boldsymbol{G}$, where $\boldsymbol{G}\in\mathbb{F}_2^{N\times N}$ is the invertible generator matrix corresponding to the coding tree for length $N$. The pseudo code for computing $\boldsymbol{G}$ is provided in Algorithm~\ref{Algorithm1}. As an example, the invertible generator matrix for BBT codes of length $N=9$ can be obtained as 
$$
\boldsymbol{G}_9=
\left(
\begin{smallmatrix}
1 & 0 & 0 & 0 & 0 & 0 & 0 & 0 & 0  \\ 
1 & 1 & 0 & 0 & 0 & 0 & 0 & 0 & 0  \\ 
1 & 0 & 1 & 0 & 0 & 0 & 0 & 0 & 0  \\ 
1 & 0 & 0 & 1 & 0 & 0 & 0 & 0 & 0  \\ 
1 & 1 & 0 & 1 & 1 & 0 & 0 & 0 & 0  \\ 
1 & 0 & 0 & 0 & 0 & 1 & 0 & 0 & 0  \\ 
1 & 1 & 0 & 0 & 0 & 1 & 1 & 0 & 0  \\ 
1 & 0 & 1 & 0 & 0 & 1 & 0 & 1 & 0  \\ 
1 & 1 & 1 & 1 & 0 & 1 & 1 & 1 & 1 \\ 
\end{smallmatrix}
\right).$$
\begin{algorithm}[t]
\caption{\emph{RecursivelyCompute$\boldsymbol{G}$}($N$)}\label{Algorithm1}
\SetKwInOut{Input}{Input}
\SetKwInOut{Output}{Output}
\Input{$N$:~a code length}
\Output{$\boldsymbol{G}$:~a BBT polar code generator matrix}
\If{$N==1$}{
$\boldsymbol{G}\leftarrow\begin{pmatrix}1\end{pmatrix}$\\return;}
$N_l\leftarrow\lceil N/2\rceil$;~$N_r\leftarrow\lfloor N/2\rfloor$\\
$\boldsymbol{G}_l \leftarrow\text{\emph{RecursivelyCompute}}\boldsymbol{G}(N_l)$;
$\boldsymbol{G}_r \leftarrow\text{\emph{RecursivelyCompute}}\boldsymbol{G}(N_r)$\\
\eIf{$N_l==N_r$}
{$\boldsymbol{G}\leftarrow
\begin{pmatrix}
\boldsymbol{G}_l & \boldsymbol{0}_{N_l\times N_r}\\
\boldsymbol{G}_r & \boldsymbol{G}_r
\end{pmatrix}$}
{$\widetilde{\boldsymbol{G}}_r \leftarrow
\begin{pmatrix}
\boldsymbol{G}~\boldsymbol{0}_{N_r\times 1}
\end{pmatrix}$;
$\boldsymbol{G}\leftarrow
\begin{pmatrix}
\boldsymbol{G}_l & \boldsymbol{0}_{N_l\times N_r}\\
\widetilde{\boldsymbol{G}}_r & \boldsymbol{G}_r
\end{pmatrix}$}
\end{algorithm}

\section{Code Construction}\label{section3}
Like the conventional polar codes, the performance of the proposed scheme is determined by the selection of the active leaf nodes. To achieve a good performance, it is critical to rank the reliabilities of the leaf nodes so that we can select the $K$ most reliable leaf nodes out of the $N$ leaf nodes. 

In this paper, we assume that a codeword is modulated using binary phase-shift keying~(BPSK) which maps $0\rightarrow+1$ and $1\rightarrow-1$, and then transmitted through an additive white Gaussian noise~(AWGN) with zero mean and variance $\sigma^2$. The SNR is defined as $E_b/N_0$, where $E_b$ is the energy per data bit and $N_0=2\sigma^2$. In this scenario, the GA method can be employed to evaluate the error performance of the leaf nodes at a given SNR. However, SNR-independent ranking is preferable for practice use. Therefore, we propose two new SNR-independent methods to rank more conveniently the leaf nodes.
\subsection{Gaussian Approximation Construction}\label{section3subsection1}
For the BBT polar codes, the reliabilities of the leaf nodes can be evaluated by the GA method~\cite{Trifonov2012} in a similar way to the conventional polar codes. We treat approximately the LLRs of the tree nodes in the decoding process as Gaussian variables and then compute the expectation of the LLRs. The GA construction of the BBT polar code is to select $K$ leaf nodes with the largest LLR expectations to transmit data bits. 

\subsection{Minimum Hamming Weight Construction}\label{section3subsection2}
Motivated by~\cite{Niu2019}, we can evaluate the reliabilities of the leaf nodes by their error performance under genie-aided SC decoding\footnote{The genie-aided SC decoder knows the correct previous leaf nodes when the current leaf node is visited.}. Given a coding tree with generator matrix $\boldsymbol{G}\in\mathbb{F}_2^{N\times N}$, for the $i$-th leaf node, define the corresponding polar subcode by
\begin{equation}
\mathscr{C}^{(i)}\triangleq\{
\boldsymbol{c}=(\boldsymbol{0}_0^{i-1},1,\boldsymbol{w}_{i+1}^{N-1})\boldsymbol{G}\mid\boldsymbol{w}_{i+1}^{N-1}\in\mathbb{F}_2^{N-i-1}\}.
\end{equation} 
The error rate of the $i$-th leaf node under genie-aided SC decoding, denoted by ${\varepsilon}^{(i)}$, can be upper bounded by
\begin{align}\label{Eq2}
{\varepsilon}^{(i)}\leq \sum_{1\leq d\leq N}A^{(i)}(d){\rm PEP}(N,d),
\end{align}
where $A^{(i)}(d)$ is the number of the codewords with Hamming weight $d$ in $\mathscr{C}^{(i)}$ and ${\rm PEP}(N,d)$ is the pairwise error probability between the all-zero codeword and a codeword with Hamming weight $d$. Based on~(\ref{Eq2}), we can give an estimate of ${\varepsilon}^{(i)}$ as
\begin{align}
{\varepsilon}^{(i)}\lesssim A^{(i)}(d_{{\rm min}}^{(i)}){\rm PEP}(N,d_{{\rm min}}^{(i)}),
\end{align}
where $d_{{\rm min}}^{(i)}$ and $A^{(i)}(d_{{\rm min}}^{(i)})$ denote the MHW of $\mathscr{C}^{(i)}$ and the number of the MHW codewords in $\mathscr{C}^{(i)}$, respectively. 

Observe that ${\varepsilon}^{(i)}$ is relevant to $d_{{\rm min}}^{(i)}$ and $A^{(i)}(d_{{\rm min}}^{(i)})$. We propose to sort the leaf nodes according to $d_{{\rm min}}^{(i)}$ in ascending order and then sort those leaf nodes with the same $d_{{\rm min}}^{(i)}$ according to $A^{(i)}(d_{{\rm min}}^{(i)})$ in descending order, resulting in an ascending reliability order of $N$ leaf nodes. The MHW construction of the BBT polar code is to select the last $K$ indices in the ordered sequence as the active set $\mathscr{A}$. 

Now we show how to evaluate $d_{{\rm min}}^{(i)}$ and $A^{(i)}(d_{{\rm min}}^{(i)})$. When $N$ is a power of two, $d_{{\rm min}}^{(i)}$ and $A^{(i)}(d_{{\rm min}}^{(i)})$ can be calculated using the iterative enumeration algorithm presented in~\cite{Niu2019}. However, when $N$ is not a power of two, it is not easy to calculate $d_{{\rm min}}^{(i)}$ and $A^{(i)}(d_{{\rm min}}^{(i)})$. Hence we turn to a randomly interleaved coding tree. To be precise, similar to~\cite{Chiu2020}, we can insert uniform interleavers between intermediate layers of a coding tree, as described below. Given a code length $N$, the root node is labeled by a codeword $\boldsymbol{c}\in\mathbb{F}_2^N$. Each node labeled by a vector $\boldsymbol{v}\in\mathbb{F}_2^{\ell}$, where $2\leq \ell\leq N$, has a left child node labeled by $\boldsymbol{v}_{l}\in\mathbb{F}_2^{\lceil \ell/2\rceil}$ and a right child node labeled by $\boldsymbol{v}_{r}\in\mathbb{F}_2^{\lfloor \ell/2\rfloor}$. Define that
\begin{align}
\boldsymbol{v}'_{r}\triangleq
\begin{cases}
\boldsymbol{v}_{r}\Pi,~&\text{if~$\lceil \ell/2\rceil=\lfloor \ell/2\rfloor$}\\
(\boldsymbol{v}_{r},0)\Pi,~&\text{if~$\lceil \ell/2\rceil=\lfloor \ell/2\rfloor+1$}
\end{cases}
,
\end{align}
where $\Pi$ is a random permutation matrix, and
\begin{align}
\boldsymbol{v}_{l}\oplus \boldsymbol{v}'_{r}=(v_{l,0}+v'_{r,0},\ldots,v'_{l,\lceil \ell/2\rceil-1}+v'_{r,\lceil \ell/2\rceil-1}),
\end{align}
where ``$+$'' denotes the addition over $\mathbb{F}_2$. The relationship between a parent node and its child nodes in the interleaved coding tree is given by 
\begin{align}
\boldsymbol{v}=(\boldsymbol{v}_{l}\oplus\boldsymbol{v}'_{r}, \boldsymbol{v}_{r}).
\end{align}
Under the assumption that all the interleavers are uniform at random, the labels of the parent nodes are random. Then, for each node with length $\ell$, we define a weight enumerating function~(WEF) as
\begin{align}
{\rm{WEF}}= \sum_{d=0}^{\ell}B(d)Y^d,
\end{align}
where $Y$ is a dummy variable and $B(d)$ is the average number of the labels with Hamming weight $d$. Let $d_{\rm min}$ be the minimum positive integer with $B(d_{\rm min}) > 0$. The minimum weight enumerating function~(MWEF) of the node is defined as
\begin{align}
{\rm{MWEF}}=B(d_{\rm min})Y^{d_{\rm min}}.
\end{align}
\begin{lemma}\label{Le1}
Consider a pair of sibling nodes in the randomly interleaved coding tree. Let $\boldsymbol{v}_{l}$ be a sample label of the left child node and $\boldsymbol{v}_{r}$ be a sample label of the right child node. The vector $\boldsymbol{v}_{l}$ is with length $\ell_l$ and Hamming weight $d_l$ and the vector $\boldsymbol{v}_{r}$ is with length $\ell_r$ and Hamming weight $d_r$. The weight distribution for $\boldsymbol{v}=(\boldsymbol{v}_{l}\oplus \boldsymbol{v}'_{r},\boldsymbol{v}_{r})$ is given by
\begin{align}
{\rm{E}}_{\Pi}\left[Y^{w_{\rm{H}}\left(\boldsymbol{v}\right)}\right] \triangleq \sum_{\Pi} \Pr\{\Pi\} Y^{w_{\rm{H}}\left(\boldsymbol{v}\right)}
=\sum_{k=\max \left(0, d_{l}+d_{r}-\ell_{l}\right)}^{\min \left(d_{l}, d_{r}\right)} \frac{
\left(\begin{array}{c}d_{l} \\ k\end{array}\right)
\left(\begin{array}{c}\ell_{l}-d_{l} \\ d_{r}-k\end{array}\right)
}{\left(\begin{array}{c}\ell_{l} \\ d_{r}\end{array}\right)} Y^{d_{l}+2d_{r}-2k},
\end{align}
where $w_{\rm{H}}\left(\boldsymbol{v}\right)$ denotes the Hamming weight of $\boldsymbol{v}$.
\end{lemma}
\begin{proof}
Let $k$ be the number of positions at which the components in $\boldsymbol{v}_{l}$ and $\boldsymbol{v}'_{r}$ are both equal to 1. With the uniform interleaver, we see that $k$ is a random integer, ranging from $\max \left(0, d_{l}+d_{r}-\ell_{l}\right)$ to $\min \left(d_{l}, d_{r}\right)$ with probability $\left(\begin{array}{c}d_{l} \\ k\end{array}\right)\left(\begin{array}{c}\ell_{l}-d_{l} \\ d_{r}-k\end{array}\right)\bigg/\left(\begin{array}{c}\ell_{l} \\ d_{r}\end{array}\right)$. Given  $k$, the Hamming weight of $\boldsymbol{v}_{l}\oplus\boldsymbol{v}'_{r}$ is $d_{l}+d_{r}-2k$. Finally, the additional concatenation of $\boldsymbol{v}_{r}$ gives the Hamming weight of $\boldsymbol{v}=(\boldsymbol{v}_{l}\oplus\boldsymbol{v}'_{r},\boldsymbol{v}_{r})$ as $d_{l}+d_{r}-2k+d_{r}=d_{l}+2d_{r}-2k$.
\end{proof}

\begin{theorem}\label{Th1}
Consider a parent node with two child nodes in the randomly interleaved coding tree. Let $\sum_{d_l=0}^{\ell_l}B_{l}(d_l)Y^{d_l}$ and $\sum_{d_r=0}^{\ell_r}B_{r}(d_r)Y^{d_r}$ be the WEF of the left child node and that of the right child node, respectively. The WEF of the parent node is
\begin{align}
{\rm{WEF}}=\sum_{d_l=0}^{\ell_l}\sum_{d_r=0}^{\ell_r}B_{l}(d_l)B_{r}(d_r)\sum_{k=\max \left(0, d_{l}+d_{r}-\ell_{l}\right)}^{\min \left(d_{l}, d_{r}\right)} \frac{
\left(\begin{array}{c}d_{l} \\ k\end{array}\right)
\left(\begin{array}{c}\ell_{l}-d_{l} \\ d_{r}-k\end{array}\right)
}{
\left(\begin{array}{c}\ell_{l} \\ d_{r}\end{array}\right)} Y^{d_{l}+2d_{r}-2k}.
\end{align}
\end{theorem}
\begin{proof}
It can be directly derived from Lemma~\ref{Le1} and the proof is omitted here.
\end{proof}
\begin{corollary}\label{Co1}
Consider a parent node with two child nodes in the randomly interleaved coding tree. If the WEF of the left child node is $1$ and the MWEF of the right node is $B_{r}(d_{{\rm min},r})Y^{d_{{\rm min},r}}$, then the MWEF of the parent node is $B_{r}(d_{{\rm min},r})Y^{2d_{{\rm min},r}}$.
\end{corollary}

\begin{proof}
Let $\sum_{d_r=0}^{\ell_r}B_{r}(d_r)Y^{d_{r}}$ be the WEF of the right child node. From Theorem~\ref{Th1}, the WEF of the parent node is $\sum_{d_r=0}^{\ell_r}B_{r}(d_r)Y^{2d_{r}}$. Then, we have that the MWEF of the parent node is $B_{r}(d_{{\rm min},r})Y^{2d_{{\rm min},r}}$.
\end{proof}

\begin{corollary}\label{Co2}
Consider a parent node with two child nodes in the randomly interleaved coding tree. If the MWEF of the left child node is $B_{l}(d_{{\rm min},l})Y^{d_{{\rm min},l}}$ and the WEF of the right node is $(Y+1)^{\ell_r}$, then the MWEF of the parent node is 
\begin{align}
B_{l}(d_{{\rm min},l})\sum_{0\leq d_{r}\leq \min\left(d_{{\rm min},l},\ell_r\right)}\frac{
\left(\begin{array}{c}\ell_{r} \\ d_{r}\end{array}\right)
\left(\begin{array}{c}d_{{\rm min},l} \\ d_{r}\end{array}\right)
}{\left(\begin{array}{c}\ell_{l} \\ d_{r}\end{array}\right)}Y^{d_{{\rm min},l}}.
\end{align}
\end{corollary}

\begin{proof}
Let $\sum_{d_l=0}^{\ell_l}B_{l}(d_l)Y^{d_l}$  be the WEF of the left child node. From Theorem~\ref{Th1}, the WEF of the parent node is 
\begin{align}
\sum_{d_{l}}B_{l}(d_{l})\sum_{0\leq d_{r}\leq \ell_r}\left(\begin{array}{c}\ell_{r} \\ d_{r}\end{array}\right)
\sum_{k=\max \left(0, d_{l}+d_{r}-\ell_{l}\right)}^{\min \left(d_{l}, d_{r}\right)} \frac{
	\left(\begin{array}{c}d_{l} \\ k\end{array}\right)
	\left(\begin{array}{c}\ell_{l}-d_{l} \\ d_{r}-k\end{array}\right)
}{
	\left(\begin{array}{c}\ell_{l} \\ d_{r}\end{array}\right)} Y^{d_{l}+2d_{r}-2k}.
\end{align}
Since $d_r\geq k$, the Hamming weight of $\boldsymbol{v}=(\boldsymbol{v}_{l}\oplus \boldsymbol{v}'_{r},\boldsymbol{v}_{r})$ is no less than the Hamming weight of $\boldsymbol{v}_l$, i.e.,   $d_l+2d_r-2k\geq d_l$. Then, the MWEF of the parent node is given by
\begin{align}
&B_{l}(d_{{\rm min},l})
\sum_{0\leq d_{r}\leq \ell_r}\left(\begin{array}{c}\ell_{r} \\ d_{r}\end{array}\right)
\sum_{k=\max \left(0, d_{{\rm min},l}+d_{r}-\ell_{l}\right)}^{\min \left(d_{{\rm min},l}, d_{r}\right)} \frac{
\left(\begin{array}{c}d_{{\rm min},l} \\ k\end{array}\right)
\left(\begin{array}{c}\ell_{l}-d_{{\rm min},l} \\ d_{r}-k\end{array}\right)
}{
\left(\begin{array}{c}\ell_{l} \\ d_{r}\end{array}\right)}\mathbb{I}\{k=d_r\}Y^{d_{{\rm min},l}}\\
=&B_{l}(d_{{\rm min},l})
\sum_{0\leq d_{r}\leq \min\left(d_{{\rm min},l},\ell_r\right)}\left(\begin{array}{c}\ell_{r} \\ d_{r}\end{array}\right)\frac{
\left(\begin{array}{c}d_{{\rm min},l} \\ d_{r}\end{array}\right)
}{
\left(\begin{array}{c}\ell_{l} \\ d_{r}\end{array}\right)}Y^{d_{{\rm min},l}},
\end{align}
where $\mathbb{I}\{k=d_r\}$ is 1 if $k=d_r$, and 0 otherwise.
\end{proof}

From Corollary~\ref{Co1} and~\ref{Co2}, we propose for polar subcodes to estimate their MHWs and their numbers of MHW codewords using the randomly interleaved coding tree. Consider the $i$-th polar subcode. To begin with, the MWEF of the $i$-th leaf node is initialized by $Y$. Next, we track the path from the $i$-th leaf node to the root over the randomly interleaved coding tree. When the MWEF of a node on the path $B(d_{\rm min})Y^{d_{\rm min}}$ is available, the MWEF of its parent node can be computed by distinguishing two cases. In the case when the node is a right child node, the MWEF of its parent node is $B(d_{\rm min})Y^{2d_{\rm min}}$, while in the case when the node is a left child node, the MWEF of its parent node is 
\begin{align}
B(d_{\rm min})
\sum_{0\leq d_{r}\leq \min\left(d_{\rm min},\ell_r\right)}
\tbinom{\ell_{r}}{d_{r}}
\frac{\tbinom{d_{\rm min}}{d_r}}{\tbinom{\ell_{l}}{d_{r}}}Y^{d_{\rm min}},
\end{align}
where $\ell_l$ is the length of the node and $\ell_r$ is the length of its sibling node. The procedure is terminated when the MWEF of the root node is calculated out. Then we can obtain the estimate of the MHW and that of the number of the MHW codewords of the $i$-th polar subcode. Hereafter, these two parameters are used to rank the leaf nodes.
\subsection{Polarization Weight Construction}\label{section3subsection3}
In the following, we propose a heuristic PW construction of the BBT polar code. 

In a coding tree, we label the branch from a parent node to its left child node by $0$ and the branch from a parent node to its right child node by $1$. Then, to each leaf node, we assign a sequence $\boldsymbol{b}=(b_0,b_1,\ldots,b_{m-1})$ that records the labels of the path~(of length $m$) from the root to the considered leaf node. This defines a one-to-one correspondence between a binary sequence and a leaf node. Notice that a leaf node at the second-to-last level is associated with a sequence of length $m=\lfloor \log_2N\rfloor$ and a leaf node at the last level is associated with a sequence of length $m=\lceil \log_2N\rceil$. In Fig.~\ref{label_of_leaf_nodes}, we show two coding trees. One is for $N=8$ and the other is for $N=6$. We have labeled the branches of both the coding trees. 
\begin{figure}[t]  
\centering
\subfloat[$N=8$]{
\includegraphics[width=0.6\textwidth]{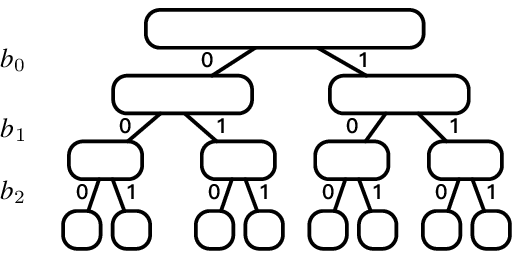}
}
\subfloat[$N=6$]{
\includegraphics[width=0.4\textwidth]{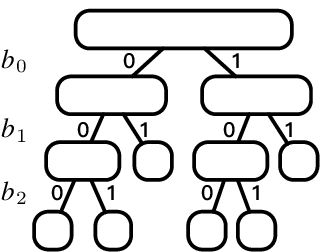}
}
\caption{The coding trees for $N=6$ and $N=8$, where the branches are labeled.}\label{label_of_leaf_nodes}
\end{figure}

For a BBT polar code, we propose to define a PW of the $i$-th leaf node node by
\begin{align}
	{\rm PW}^{(i)}\triangleq \sum_{j=0}^{m-1}b_j\cdot\kappa^{J-j},
\end{align}
where $\kappa$ is a design parameter and $J\triangleq\lceil\log_2N\rceil-1$. When $N$ is a power of two, all the leaf nodes are associated with the sequences with length $m=\log_2N$, which are exactly the binary expressions of the indices of the leaf nodes~(see Fig.~\ref{label_of_leaf_nodes}(a) for reference), indicating that the proposed PW for the BBT polar code is consistent with the PW introduced in~\cite{He2017} for the conventional polar code. In this paper, the parameter $\kappa=2^{1/4}$, as proposed in~\cite{3GPP167209}. The PW construction of the BBT polar code is to select $K$ leaf nodes with the largest PWs to transmit data bits. 
\subsection{Performance Evaluation and Latency Analysis}\label{section3subsection4}
\subsubsection{Decoding performance}
Fig.~\ref{BBTPCCs_N384K192} and Fig.~\ref{BBTPCCs_N768K384} show the performance comparisons among the proposed three BBT polar code constructions for $N=384$ and $K=192$ and for $N=768$ and $K=384$, respectively. The design SNR for the GA construction is $3~{\rm dB}$. The SC deocding, the SCL decoding with list size 8~(denoted by SCL(8)) and the CA-SCL decoding with list size 8 and 11-bit CRC~(denoted by CA-SCL(11,8)) are implemented in the simulations. 
\begin{figure}[t]
	\vskip-3cm
	\centering
	\includegraphics[width=\textwidth]{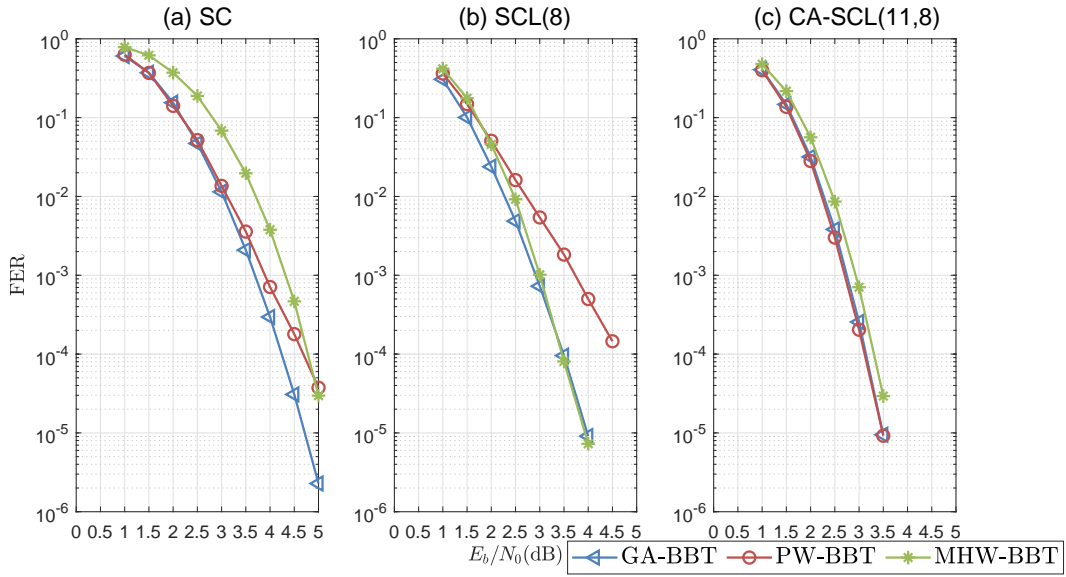}
	\caption{Performance comparisons among the proposed three BBT polar code constructions under the SC decoding, the SCL decoding with list size 8 and the CA-SCL decoding with list size 8 and 11-bit CRC. Here, $N=384$ and $K=192$.}\label{BBTPCCs_N384K192}
\end{figure}
\begin{figure}[t]
	\vskip-2cm
	\centering
	\includegraphics[width=\textwidth]{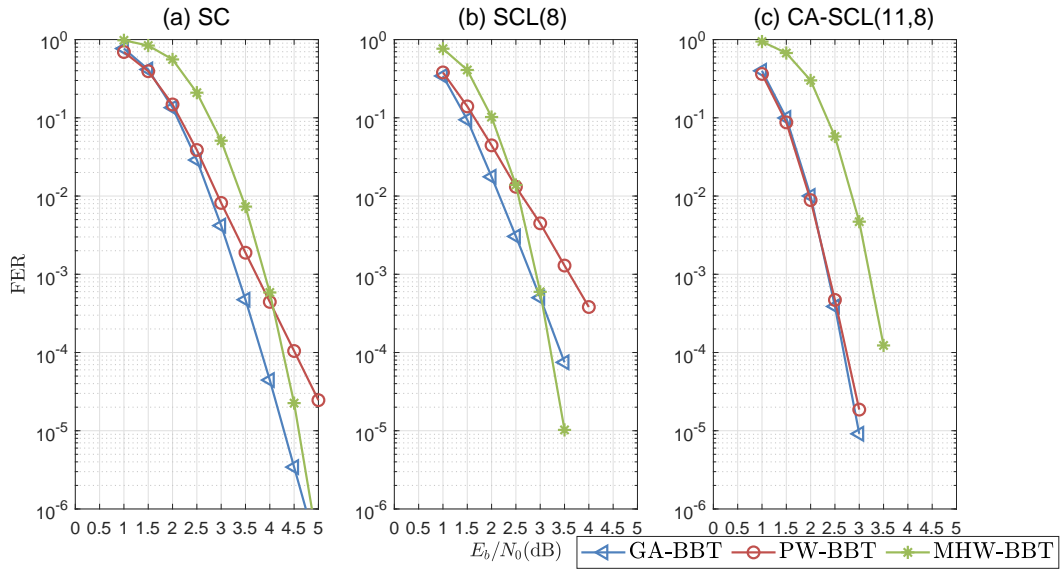}
	\caption{Performance comparisons among the proposed three BBT polar code constructions under the SC decoding, the SCL decoding with list size 8 and the CA-SCL decoding with list size 8 and 11-bit CRC. Here, $N=768$ and $K=384$.}\label{BBTPCCs_N768K384}
\end{figure}
\clearpage
We observe that 
\begin{enumerate}
\item Under the SC decoding and the SCL decoding, the GA construction performs the best in the low SNR region. However, as SNR increases in the high SNR region, the MHW construction tends to perform the best.
\item Under the CA-SCL decoding, the performance of the PW construction is similar to that of the GA construction. The MHW construction performs worse than the GA construction and the PW construction. A possible reason is that, in the case of the CA-SCL decoding, the actual code rate of a BBT polar code increases, and the MHW construction needs a higher SNR to achieve satisfactory performance.
\end{enumerate}

For practical use, length-flexible polar coding schemes have been devised in 5G~\cite{3GPP38212}. In this paper, we take the polar coding scheme for 5G downlink scenarios for comparison. In our simulations of 5G polar codes, a data sequence is first encoded by a CRC code. Before polar encoding, one of 5G rate-matching schemes~(puncturing, shortening and repetition) is selected according to code parameters, and a pre-freezing step is performed. Then, an output of the CRC encoder is encoded by a polar encoder, resulting in a polar codeword. Eventually, the codeword is sent to a sub-block interleaver and some bits of the interleaved codeword are chosen by a circular buffer to transmit. Note that for an unconstrained data length, we neglect 5G input bit interleaver in our simulations.

Fig.~\ref{BBTvs5G_N768} shows the performance comparisons between the constructed BBT polar codes and the 5G polar codes. The CA-SCL decoding with list size 8 and 11-bit CRC is implemented in the simulations. The tested code length is $N=768$ and code rates are $R\in\{1/4,1/2,3/4\}$. We observe that for the low and moderate rates, the PW-BBT polar codes can achieve similar performance to the 5G polar codes.

\begin{figure}[t]
	\centering
	\includegraphics[width=0.65\textwidth]{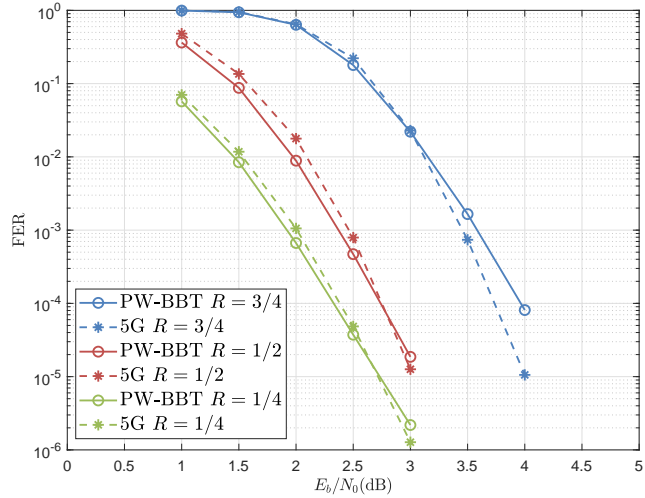}
	\caption{Performance comparisons between the constructed BBT polar codes and the 5G polar codes. The CA-SCL decoding with list size 8 and 11-bit CRC is implemented in the simulations. Here, $N=768$ and $R\in\{1/4,1/2,3/4\}$. }\label{BBTvs5G_N768}
\end{figure}

\subsubsection{Decoding complexity and latency}
We analyze the SC decoding latency of our proposed codes by examining the number of LLR calculations defined in~(\ref{Equation4}) and (\ref{Equation5}). To decode a single codeword of length $N$, the SC decoder is implemented over a coding tree which has $\lceil \log_2N\rceil$ levels. For each level, the decoder executes at most $N$ LLR calculations. The total number of  LLR calculations in the SC decoding process required in our scheme is at most $N\lceil \log_2N\rceil$, which~(in the case of when $N$ is not a power of two) is less than that required in the puncturing and shortening polar coding schemes, namely $2^{\lceil \log_2N\rceil}\lceil \log_2N\rceil$. As a concrete example, for $N = 768$, in our scheme, the number is less than 7680, while in puncturing and shortening polar coding schemes, the number is 10240. In this sense, the decoding latency~(as well as complexity) can be reduced by about 25$\%$.% ((2^[logN])-N) / (2^[logN]) ?

\section{Partitioned SC-based Decoding}\label{section4}
\subsection{Decoding Sub-tree}\label{section4subsection1}
\begin{figure}[t]
	\vskip-1.5cm
	\centering
	\includegraphics[width=0.48\textwidth]{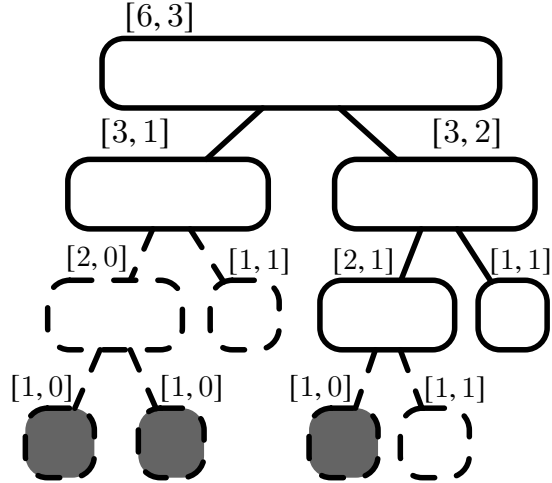}
	\caption{The decoding sub-tree~(solid line) for $N=6$ and $\tau=1$, which is extracted from the coding tree whose descendants involve 3 active leaf nodes and 3 frozen leaf nodes~(marked in gray). Here, each node is labeled by its associated code $[n, k]$.}\label{decoding_subtree}
	\vskip-0.5cm
\end{figure}

Recall that for the BBT polar codes, both the encoding and the SC-based decoding can be implemented over the coding tree defined in Section~\ref{section2subsection1}. However, the SC-based decoder makes hard decisions at leaf nodes in a serial manner, which results in high decoding latency when the code length is long. In the following, we define a decoding sub-tree and then present a low-latency SC-based decoding algorithm.

Now assume that a rate-profiling is given for a coding tree of length $N$. Each node in the coding tree is associated with a binary linear code, consisting of all the possible labels of the node. Specifically, for a tree node of length $n$ whose descendants involve $k$ active leaf nodes and $n-k$ frozen leaf nodes, the associated code has length $n$ and dimension $k$. For convenience, we say that such a node has a dimension $k$. When $k$ is small, it is feasible to compute all the codewords in the code. Let $\tau$~(a positive integer) be a dimension threshold. We first identify all the nodes which has dimension at most $\tau$ and whose parent node has dimension greater than $\tau$ in the coding tree and then extract a decoding sub-tree from the coding tree by removing all the descendants of these identified nodes. We refer to the leaves in the decoding sub-tree as the decoding leaves. An example of a decoding sub-tree for $N=6$ and $\tau=1$ is shown in Fig.~\ref{decoding_subtree}. We see that, the decoding sub-tree has two levels~(less than the original coding tree) and three decoding leaves.

Suppose that there are $q$ decoding leaves in the decoding sub-tree. Denote by $\mathcal{N}_{\rm{dec}}^{(t)}$ the $t$-th leaf from left to right in the decoding sub-tree and by $\mathscr{S}^{(t)}$ the code associated with $\mathcal{N}^{(t)}_{\rm{dec}}$, where $t=0,1,\ldots,q-1$. Then, a juxtaposition of the labels of the $q$ leaves, written as $\{\boldsymbol{v}^{(t)}\}_{t=0}^{q-1}\triangleq(\boldsymbol{v}^{(0)},\boldsymbol{v}^{(1)},\ldots,\boldsymbol{v}^{(q-1)})$, can be transformed recursively into a BBT polar codeword at the root, where $\boldsymbol{v}^{(t)}\in \mathscr{S}^{(t)}$ has length $\ell^{(t)}$. Note that the lengths $\ell^{(t)}(t=0,1,\ldots,q-1)$ may be different. 

When the LLRs of a decoding leaf $\mathcal{N}^{(t)}_{\rm{dec}}$, denoted by $\boldsymbol{\alpha}^{(t)}\triangleq (\alpha^{(t)}_0,\alpha^{(t)}_1,\ldots,\alpha^{(t)}_{\ell^{(t)}-1})$, are available, the partitioned SC~(PSC) decoder computes the likelihoods of all the codewords in $\mathscr{S}^{(t)}$ and then directly sets the HBEs of $\mathcal{N}^{(t)}_{\rm{dec}}$, denoted by $\boldsymbol{\beta}^{(t)}\triangleq(\beta^{(t)}_0,\beta^{(t)}_1,\ldots,\beta^{(t)}_{\ell^{(t)}-1})$,  to the most likely codeword. The partitioned SCL~(PSCL)  decoder performs in a similar way. Define a decoding path at decoding step $t$ as a juxtaposition of the HBEs of the previous $t$ decoding leaves, i.e., $\{\boldsymbol{\beta}^{(i)}\}_{i=0}^{t-1}\triangleq(\boldsymbol{\beta}^{(0)},\boldsymbol{\beta}^{(1)},\ldots,\boldsymbol{\beta}^{(t-1)})$. In the PSCL decoding, at most $L$ decoding paths are retained. At the $(t+1)$-th decoding step, the LLRs of the leaf $\mathcal{N}^{(t)}_{\rm{dec}}$ are available. The PSCL decoder splits each path $\{\boldsymbol{\beta}^{(i)}\}_{i=0}^{t-1}$ into at most $2^{\tau}$ paths by setting $\boldsymbol{\beta}^{(t)}$ to all the codewords in $\mathscr{S}^{(t)}$, and for each extended path $\{\boldsymbol{\beta}^{(i)}\}_{i=0}^{t}$, computes its path metric as 
\begin{align}
{\rm PM}_{t}=\sum_{i=0}^{t}\sum_{j=0}^{\ell^{(i)}-1}\ln(1+\exp(-(1-2\beta^{(i)}_{j})\alpha^{(i)}_{j})).
\end{align}
Then, the $L$ paths with the minimum path metrics are retained, while the other paths are discarded. Eventually, the path with the minimum path metric is selected as the decoder output. 
\subsection{Performance Analysis}\label{section4subsection2}
Without loss of generality, we assume that the all-zero codeword is transmitted over a BPSK-AWGN channel, resulting in a received vector. In this case, all the labels of the decoding leaves are zero vectors.

Similar to~\cite[Lemma~1]{Mori2009ISIT}, it can be easily verified that the performance of the PSC decoding is identical to that of the genie-aided PSC decoding where the decoder knows the correct previous decoding leaves when the current decoding leaf is visited, and the error event under the genie-aided PSC decoding is a union of the events that the HBEs of the decoding leaves are incorrect. Then, for a decoding sub-tree having $q$ leaves $\mathcal{N}_{\rm{dec}}^{(t)}$, $t=0,1,\ldots,q-1$, the FER of the PSC decoding is given by
\begin{align}
{\rm FER }=\Pr\big\{\bigcup_{t=0}^{q-1}E^{(t)}\big\},
\end{align}
where $E^{(t)}$ denotes the error event for $\mathcal{N}_{\rm{dec}}^{(t)}$ under the genie-aided PSC decoding, and which can be upper bounded as
\begin{align}\label{UBFER}
{\rm FER }\leq\sum_{t=0}^{q-1}\Pr\{E^{(t)}\},
\end{align}
and lower bounded as
\begin{align}\label{LBFER}
{\rm FER }\geq\max_{0\leq t\leq q-1}\Pr\{E^{(t)}\}.
\end{align}

To analyze the ML decoding error probability $\Pr\{E^{(t)}\}$ of a decoding leaf $\mathcal{N}_{\rm{dec}}^{(t)}$, we first define its ``channel'' transition probability as
\begin{align}
P(\boldsymbol{y},\{\boldsymbol{v}^{(i)}\}_{i=0}^{t-1}|\boldsymbol{v}^{(t)})\triangleq\sum_{\{\boldsymbol{v}^{(i)}\}_{i=t+1}^{q-1}\in\{\mathbb{F}_2^{\ell^{(i)}}\}_{i=t+1}^{q-1}}\frac{1}{2^{N-\ell^{(t)}}}\Pr\{\boldsymbol{y}|\{\boldsymbol{v}^{(i)}\}_{i=0}^{q-1}\},
\end{align}
where $\boldsymbol{y}\in\mathbb{R}^N$ and $\boldsymbol{v}^{(i)}\in\mathbb{F}_2^{\ell^{(i)}},i=0,1,\ldots,q-1$. Recall that $\mathcal{N}_{\rm{dec}}^{(t)}$ is associated with the code $\mathscr{S}^{(t)}$. Let $\boldsymbol{v}_{0}^{(t)}$ be the all-zero codeword in $\mathscr{S}^{(t)}$ and $\boldsymbol{v}_{i}^{(t)},i=1,2,\ldots,|\mathscr{S}^{(t)}|-1$, be the non-zero codewords in $\mathscr{S}^{(t)}$. The Hamming weight of $\boldsymbol{v}_{i}^{(t)}$ is denoted by $w_{\rm H}(\boldsymbol{v}_{i}^{(t)})$. Define by $E^{(t)}_{0\rightarrow i}$ the event that $\boldsymbol{v}_{i}^{(t)}$ is not less likely than $\boldsymbol{v}_{0}^{(t)}$, i.e.,
\begin{equation}
	E^{(t)}_{0\rightarrow i}\triangleq
	\{\boldsymbol{y} \in \mathbb{R}^N \mid P(\boldsymbol{y},\{\mathbf{0}^{(i)}\}_{i=0}^{t-1} \mid \boldsymbol{v}_{i}^{(t)}) \geq P(\boldsymbol{y},\{\mathbf{0}^{(i)}\}_{i=0}^{t-1} \mid \boldsymbol{v}_{0}^{(t)})\}.
\end{equation}
Based on the Bonferroni inequalities~\cite{Hoppe1985}, we have
\begin{equation}\label{1stBonferroniInequality}
	\Pr\{E^{(t)}\}\leq\sum_{1\leq i \leq |\mathscr{S}^{(t)}|-1}\Pr\{E^{(t)}_{0\rightarrow i}\},
\end{equation}
and
\begin{equation}\label{2ndBonferroniInequality}
	\Pr\{E^{(t)}\}\geq\sum_{1\leq i \leq |\mathscr{S}^{(t)}|-1}\Pr\{E^{(t)}_{0\rightarrow i}\}-\sum_{1\leq i < j \leq |\mathscr{S}^{(t)}|-1 }\Pr\{E^{(t)}_{0\rightarrow i}\cap E^{(t)}_{0\rightarrow j}\}.
\end{equation}

Since the BPSK-AWGN channel is memoryless, the transition probability associated with $\mathcal{N}_{\rm{dec}}^{(t)}$ can be written as a product of the transition probabilities associated with the components of $\mathcal{N}_{\rm{dec}}^{(t)}$. Denote by $P(\boldsymbol{y},\{\boldsymbol{v}^{(i)}\}_{i=0}^{t-1}|v_j^{(t)})$ the transition probability associated with the $j$-th component of $\mathcal{N}_{\rm{dec}}^{(t)}$. We have
\begin{align}
P(\boldsymbol{y},\{\boldsymbol{v}^{(i)}\}_{i=0}^{t-1}|\boldsymbol{v}^{(t)})=\prod_{j=0}^{\ell^{(t)}-1}P(\boldsymbol{y},\{\boldsymbol{v}^{(i)}\}_{i=0}^{t-1}|v_j^{(t)}).
\end{align} 
That is, the vector channel associated with $\mathcal{N}_{\rm{dec}}^{(t)}$ can be split into $\ell^{(t)}$ different and independent component channels. Note that the ML decoding performance over the vector channel is worse than that over $\ell^{(t)}$ uses of the best component channel and better than that over $\ell^{(t)}$ uses of the worst component channel. Then, we present two upper bounds according to the worst component channel and a lower bound according to the best component channel as below.

By using the GA method~\cite{Trifonov2012}, we can estimate the noise variance of the worst component channel, denoted by $(\sigma_{\rm worst}^{(t)})^2$, and that of the best component channel, denoted by $(\sigma_{\rm best}^{(t)})^2$. Then, according to~\cite{Sason2006}, an upper bound on the error probability $\Pr\{E^{(t)}\}$ can be obtained by 
\begin{equation}\label{ub1}
\Pr\{E^{(t)}\}\leq\sum_{1\leq i \leq |\mathscr{S}^{(t)}|-1}Q(\sqrt{w_{\rm H}(\boldsymbol{v}_{i}^{(t)})}\bigg/\sigma_{\rm worst}^{(t)}),
\end{equation}
and a lower bound on the error probability $\Pr\{E^{(t)}\}$ can be obtained by 
\begin{align}\label{lb}
\Pr\{E^{(t)}\}
\notag \geq&\sum_{1\leq i \leq |\mathscr{S}^{(t)}|-1}Q(\sqrt{w_{\rm H}(\boldsymbol{v}_{i}^{(t)})}\bigg/\sigma_{\rm best}^{(t)})\\
&~~~-\sum_{1\leq i < j \leq |\mathscr{S}^{(t)}|-1 }\psi(\rho_{ij},\sqrt{w_{\rm H}(\boldsymbol{v}_{i}^{(t)})}\bigg/\sigma_{\rm best}^{(t)},\sqrt{w_{\rm H}(\boldsymbol{v}_{j}^{(t)})}\bigg/\sigma_{\rm best}^{(t)}),
\end{align}
where 
\begin{equation}
	Q(x) \triangleq \frac{1}{\sqrt{2 \pi}} \int_x^{\infty} e^{-t^2 / 2} \mathrm{~d} t,
\end{equation}
\begin{align}
	\psi(\rho,x, y)= 
	\notag& \frac{1}{2 \pi} \int_0^{\frac{\pi}{2}-\tan ^{-1}\left(\frac{y}{x}\right)} \frac{\sqrt{1-\rho }}{1-\rho \sin 2 \theta} \exp \left[-\frac{x^2}{2} \frac{1-\rho \sin 2 \theta}{\left(1-\rho^2\right) \sin ^2 \theta}\right] {\rm d} \theta\\
	&+ \frac{1}{2 \pi} \int_0^{\tan ^{-1}\left(\frac{y}{x}\right)} \frac{\sqrt{1-\rho}}{1-\rho \sin 2 \theta} \exp \left[-\frac{y^2}{2} \frac{1-\rho \sin 2 \theta}{\left(1-\rho^2\right) \sin ^2 \theta}\right] {\rm d} \theta
\end{align}
and
\begin{equation}
	\rho_{i j}=\frac{w_{\rm H}(\boldsymbol{v}_{i}^{(t)})+w_{\rm H}(\boldsymbol{v}_{j}^{(t)})-w_{\rm H}(\boldsymbol{v}_{i}^{(t)}-\boldsymbol{v}_{j}^{(t)})}{2 \sqrt{w_{\rm H}(\boldsymbol{v}_{i}^{(t)}) w_{\rm H}(\boldsymbol{v}_{j}^{(t)})}}.
\end{equation}
In addition, by recursively computing the Bhattacharyya bounds~\cite{Arikan2009}, we can obtain the Bhattacharyya upper bound of the worst component channel, denoted by $Z_{\rm worst}^{(t)}$. The error probability $\Pr\{E^{(t)}\}$ can also be upper bounded by 
\begin{equation}\label{ub2}
\Pr\{E^{(t)}\}\leq\sum_{1\leq i \leq |\mathscr{S}^{(t)}|-1}(Z_{\rm worst}^{(t)})^{w_{\rm H}(\boldsymbol{v}_{i}^{(t)})}.
\end{equation}

Finally, by substituting (\ref{ub1}),~(\ref{ub2}) and (\ref{lb}) into (\ref{UBFER}) and (\ref{LBFER}), we can derive the following bounds on the FER of the BBT polar code with the PSC decoding.
\begin{enumerate}
\item The upper bound based on GA~(G-UB):
\begin{equation}\label{GUB}
	{\rm FER }\leq\sum_{t=0}^{q-1}\sum_{1\leq i \leq |\mathscr{S}^{(t)}|-1}Q(\sqrt{w_{\rm H}(\boldsymbol{v}_{i}^{(t)})}\bigg/\sigma_{\rm worst}^{(t)}).
\end{equation}
\item The upper bound based on Bhattacharyya bounds~(B-UB):
\begin{equation}\label{BUB}
	{\rm FER }\leq\sum_{t=0}^{q-1}\sum_{1\leq i \leq |\mathscr{S}^{(t)}|-1}(Z_{\rm worst}^{(t)})^{w_{\rm H}(\boldsymbol{v}_{i}^{(t)})}.
\end{equation}
\item The lower bound~(LB):
\begin{align}\label{LB}
	{\rm FER }\geq\max_{0\leq t\leq q-1}
\{
\notag&\sum_{1\leq i \leq |\mathscr{S}^{(t)}|-1}Q(\sqrt{w_{\rm H}(\boldsymbol{v}_{i}^{(t)})}\bigg/\sigma_{\rm best}^{(t)})\\
&~~~-\sum_{1\leq i < j \leq |\mathscr{S}^{(t)}|-1 }\psi(\rho_{ij},\sqrt{w_{\rm H}(\boldsymbol{v}_{i}^{(t)})}\bigg/\sigma_{\rm best}^{(t)},\sqrt{w_{\rm H}(\boldsymbol{v}_{j}^{(t)})}\bigg/\sigma_{\rm best}^{(t)})\}.
\end{align}
\end{enumerate}

\subsection{Performance Evaluation and Latency Analysis}\label{section4subsection3}
\subsubsection{Decoding performance}
\begin{figure}[t]  
	\vskip-2cm
	\centering
	\includegraphics[width=0.68\textwidth]{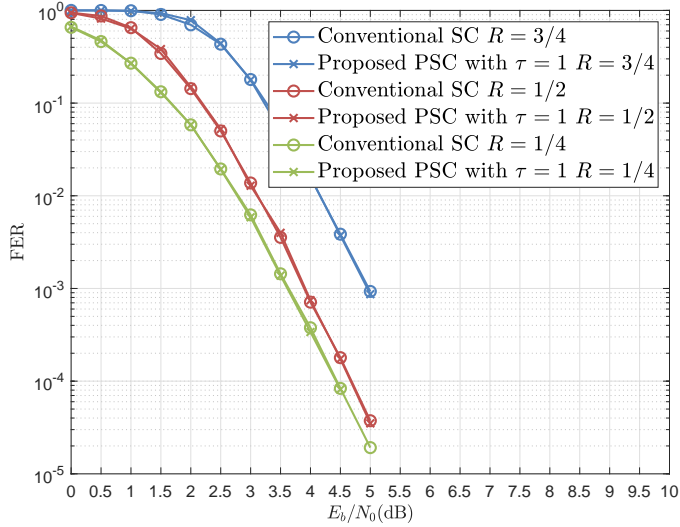}
	\caption{Performance comparisons between the proposed PSC decoding with $\tau=1$ and the conventional SC decoding. Here, $N=384$ and $R\in\{1/4,1/2,3/4\}$.}\label{CSCvsPSC_N384}
\end{figure}
\begin{figure}[t]  
	\vskip-2cm
	\centering
	\includegraphics[width=0.68\textwidth]{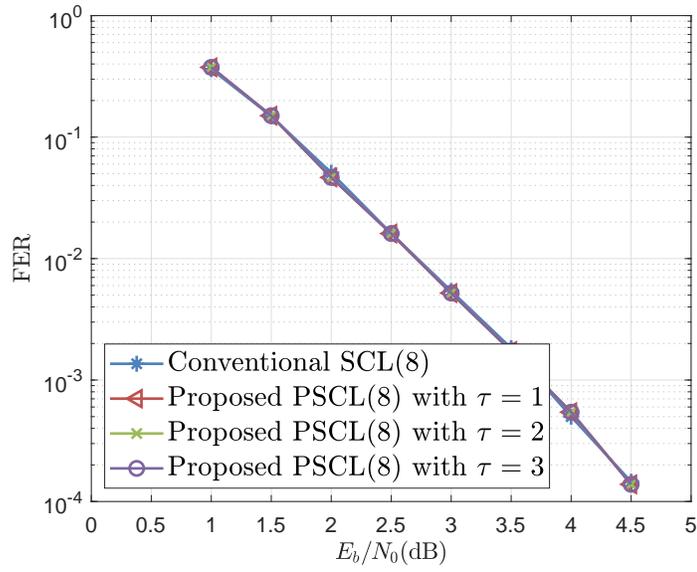}
	\vskip-0.5cm
	\caption{Performance comparisons among the proposed PSCL decoding with $\tau\in\{1,2,3\}$ and the conventional SCL decoding, where the list size is 8. Here, $N=384$ and $K=192$.}\label{CSCL8vsPSCL8_N384K192}
\end{figure}
\begin{figure}[t]  
	\vskip-1cm
	\centering
	\includegraphics[width=0.68\textwidth]{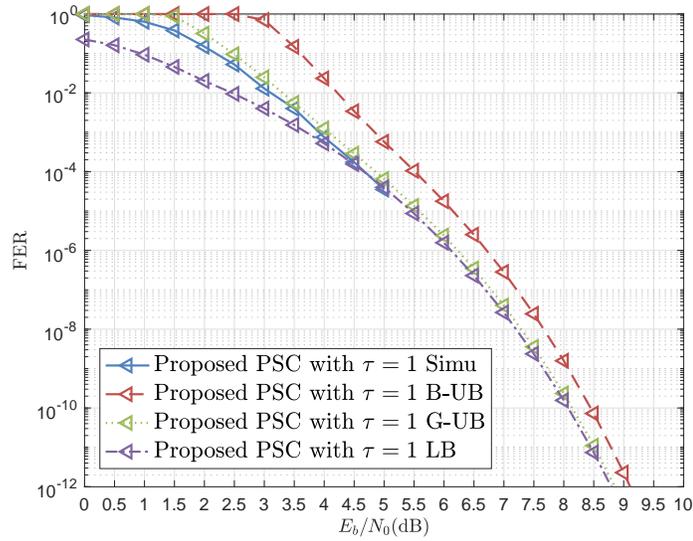}
	\vskip-0.5cm
	\caption{Performance bounds of the proposed PSC decoding with $\tau=1$. Here, $N=384$ and $R=1/2$.}\label{CSCvsPSCvsPSCUB1vsPSCUB2vsPSCLB_N384K192}
	\vskip-1cm
\end{figure}
We conduct simulations for the BBT polar codes constructed by the PW method introduced in Subsection~\ref{section3subsection3} over a BPSK-AWGN channel. For $N=384$ and $R\in\{1/4,1/2,3/4\}$, the performance comparisons between the proposed PSC decoding with $\tau=1$ and the conventional SC decoding are shown in Fig.~\ref{CSCvsPSC_N384}. For $N=384$ and $R=1/2$, the performance comparisons between the proposed PSCL(8) decoding with $\tau\in\{1,2,3\}$ and the conventional SCL(8) decoding are shown in Fig.~\ref{CSCL8vsPSCL8_N384K192}. We see that the proposed PSC/PSCL(8) decoding can have the same performance as the conventional SC/SCL(8) decoding.  However, the PSC-based decoding avoids a complete traversal of a coding tree and can reduce the decoding latency. The G-UB~(\ref{GUB}), B-UB~(\ref{BUB}) and LB~(\ref{LB}) of the PSC decoding with $\tau=1$ are shown in Fig.~\ref{CSCvsPSCvsPSCUB1vsPSCUB2vsPSCLB_N384K192}. We see that, for $\tau=1$, the G-UB is close to the LB in the high SNR region. We also see that the B-UB is looser than the G-UB. However, calculating the B-UBs is simpler.
\begin{table}[t]
	\centering
	\caption{The numbers of LLR calculations for $N=384$}
	\begin{tabular}{llll}
		\hline
		R & 1/4 & 1/2 & 3/4 \\ \hline
		SC & 3328 & 3328 & 3328 \\ 
		PSC, $\tau=1$ & 1965 & 2586 & 3023 \\ 
		PSC, $\tau=2$ & 1674 & 2322 & 2778 \\ 
		PSC, $\tau=3$ & 1602 & 2148 & 2490 \\ \hline
	\end{tabular}
	\label{The_number_of_LLR_calculations}
\end{table}

\subsubsection{Decoding complexity and latency}
We compare the decoding latency of the proposed PSC decoder and the conventional SC decoder by examining the number of LLR calculations defined in~(\ref{Equation4}) and (\ref{Equation5}). Table~\ref{The_number_of_LLR_calculations} shows the numbers of the LLR calculations required to decode BBT polar codes with $N=384$ and $R\in\{1/4,1/2,3/4\}$ for decoding. We see that the PSC decoder requires less LLR calculations than the conventional SC decoder. Moreover, for the PSC decoding, the larger the dimension threshold $\tau$ is, the less the number of LLR calculations is.

\section{Conclusion}\label{section5}
In this paper, we have defined a coding tree and then proposed a length-flexible polar coding scheme. For code construction, we have presented an SNR-dependent and two SNR-independent methods. Numerical results show that the error-rate performance of the proposed BBT polar codes is comparable to that of the 5G polar codes. Moreover, in our scheme, the SC-based decoding can be implemented without increasing the decoding latency. Further, we have also proposed the partitioned SC-based decoding, which can be implemented over the decoding sub-tree. Especially, we have derived three bounds to evaluate the performance of the PSC decoding. Numerical results show that the proposed PSC-based decoding performs as well as the conventional SC-based decoding but has less complexity and low latency.
\ifCLASSOPTIONcaptionsoff
\newpage
\fi
\bibliographystyle{IEEEtran}
\bibliography{IEEEabrv,RefYaoxym}
\end{document}